\documentclass{article}

\usepackage{arxiv}
\usepackage{url}
\usepackage{fancyhdr}
\usepackage{multirow}
\usepackage{diagbox}
\usepackage[T1]{fontenc}

\RequirePackage{algorithm}
\usepackage[backref]{hyperref}

\usepackage{bibentry}


\usepackage{xcolor}
\definecolor{lightgray}{gray}{0.893}
\usepackage{colortbl}

\usepackage{xcolor}         
\usepackage{bm}
\usepackage{amsmath}
\usepackage{amsthm} 
\usepackage{multicol} 

\usepackage{subcaption}

\usepackage{algorithm}
\usepackage{algorithmic}
\usepackage{soul}
\usepackage{booktabs}

\usepackage{amssymb}
\usepackage{bm}
\usepackage{amsmath}
\usepackage{amsthm}
\newtheorem{assumption}{Assumption}
\newtheorem{theorem}{Theorem}
\newtheorem{lemma}{Lemma}
\usepackage{times}
\usepackage{helvet}
\usepackage{courier}
\usepackage{xcolor}
\usepackage{booktabs}
\usepackage{multirow}

\usepackage{graphicx} 

\title{Multi-agent In-context Coordination via Decentralized Memory Retrieval}

\author{
Tao Jiang$^{1,2}$, Zichuan Lin$^{3}$\thanks{Joint Corresponding Authors},
Lihe Li$^{1,2}$, Yi-Chen Li$^{1,2}$, Cong Guan$^{1,2}$, 
\\ \textbf{Lei Yuan$^{1,2}$, Zongzhang Zhang$^{1,2}$\footnotemark[1], Yang Yu$^{1,2}$, Deheng Ye$^{3}$}\\   
$^1$ National Key Laboratory of Novel Software Technology, Nanjing University, Nanjing, China\\
$^2$ School of Artificial Intelligence, Nanjing University, Nanjing, China\\
$^3$ Tencent, Shenzhen, China\\
\texttt{\{jiangt,lilh,liyc,guanc,yuanl\}@lamda.nju.edu.cn}, \texttt{\{zzzhang, yuy\}@nju.edu.cn},\\
\texttt{\{zichuanlin, dericye\}@tencent.com}
}

\begin{document}
\date{}
\maketitle
\begin{abstract}
Large transformer models, trained on diverse datasets, have demonstrated impressive few-shot performance on previously unseen tasks without requiring parameter updates. This capability has also been explored in Reinforcement Learning (RL), where agents interact with the environment to retrieve context and maximize cumulative rewards, showcasing strong adaptability in complex settings. However, in cooperative Multi-Agent Reinforcement Learning (MARL), where agents must coordinate toward a shared goal, decentralized policy deployment can lead to mismatches in task alignment and reward assignment, limiting the efficiency of policy adaptation. To address this challenge, we introduce \textbf{M}ulti-\textbf{A}gent \textbf{I}n-\textbf{C}ontext \textbf{C}oordination via Decentralized Memory Retrieval (MAICC), a novel approach designed to enhance coordination by fast adaptation. Our method involves training a centralized embedding model to capture fine-grained trajectory representations, followed by decentralized models that approximate the centralized one to obtain team-level task information. Based on the learned embeddings, relevant trajectories are retrieved as context, which, combined with the agents' current sub-trajectories, inform decision-making. During decentralized execution, we introduce a novel memory mechanism that effectively balances test-time online data with offline memory. Based on the constructed memory, we propose a hybrid utility score that incorporates both individual- and team-level returns, ensuring credit assignment across agents. Extensive experiments on cooperative MARL benchmarks, including Level-Based Foraging (LBF) and SMAC (v1/v2), show that MAICC enables faster adaptation to unseen tasks compared to existing methods. Code is available at \url{https://github.com/LAMDA-RL/MAICC}.
\end{abstract}

\section{Introduction}

In-Context Learning (ICL) has emerged as a compelling paradigm for few-shot generalization, enabling models to tackle novel tasks by interpreting contextual cues without explicit retraining~\cite{icl}. This approach is epitomized by large language models, whose remarkable in-context abilities, unlocked through pretraining on vast web-scale corpora, have set new standards in natural language processing~\cite{survey_icl}. The success of this paradigm has catalyzed a parallel pursuit within Reinforcement Learning (RL) to instill agents with similar on-the-fly policy adaptation capabilities~\cite{survey_icrl}. To this end, the prevailing strategy reformulates RL as a sequence modeling problem~\cite{dt}: agents are trained on diverse trajectory datasets to internalize learning algorithms, allowing them to adapt to novel downstream tasks by conditioning on a few contextual examples~\cite{ad,dpt}. This burgeoning field of In-Context Reinforcement Learning (ICRL) has shown notable promise, but its success has mainly been demonstrated in structured environments such as single-agent grid worlds and game-based tasks.

ICRL has demonstrated strong capabilities for fast adaptation in single-agent environments. Typically, these methods condition on in-context trajectories and maintain a memory that is continuously updated with new online experiences to inform decision-making. Despite its notable success, extending this paradigm to cooperative Multi-Agent Reinforcement Learning (MARL) scenarios presents significant challenges. Unlike the single-agent setting, where an agent aims to maximize its individual cumulative rewards, cooperative MARL requires multiple agents to collaborate towards a shared objective~\cite{survey_marl}. This collaborative nature introduces distinct challenges, particularly when deployed in a decentralized manner~\cite{ctde2016}. Firstly, decentralized execution confines each agent to its local observations, often leading to a biased or incomplete understanding of the overall task characteristics. Secondly, agents typically receive only a shared team-level reward, making it difficult to assess individual contributions. This ambiguity in credit assignment can lead to the ``lazy agent" problem~\cite{vdn}, where certain agents fail to learn effective policies and contribute meaningfully to the team's success. These twin challenges of partial observability and credit assignment critically undermine the efficacy of conventional ICRL approaches in the MARL setting. Therefore, given the proven adaptive capabilities of ICL, a method that enables efficient adaptation to unseen cooperative tasks in decentralized multi-agent settings is urgently needed.

To address the above objective, we propose \textbf{M}ulti-\textbf{A}gent \textbf{I}n-\textbf{C}ontext \textbf{C}oordination via Decentralized Memory Retrieval (MAICC), a framework designed for rapid team coordination under Decentralized Partially Observable Markov Decision Processes (Dec-POMDPs)~\cite{decpomdp}. Specifically, we train a single centralized embedding model to extract fine-grained trajectory representations, and multiple decentralized embedding models for decentralized execution that approximate the centralized model to obtain team-level task information. With these pretrained models, we retrieve relevant multi-agent trajectories to serve as in-context examples. These retrieved trajectories, combined with agents’ current sub-trajectories, are used to guide and improve the decision-making process. During test time, we introduce a novel memory mechanism that efficiently balances an online replay buffer with a multi-task offline dataset for trajectory retrieval. Building upon this memory, we design a hybrid utility score that integrates both individual- and team-level returns, thereby enabling more accurate credit assignment across agents.

We evaluate our method on several standard cooperative MARL benchmarks, including Level-Based Foraging (LBF)~\cite{lbf} and the StarCraft Multi-Agent Challenge (SMAC) v1 and v2~\cite{smac,smacv2}.  Experimental results show that MAICC, equipped with efficient trajectory retrieval for ICL, enables significantly faster adaptation to unseen tasks compared to existing ICRL and multi-task MARL baselines. Additionally, visualizations of the learned trajectory embeddings verify the effectiveness of our embedding model design, capturing both individual- and team-level behavior patterns. Ablation studies further isolate and confirm the contribution of each key component in our framework. Together, these findings demonstrate MAICC’s strong empirical performance, its ability to address the limitations of prior ICRL approaches, and its potential for broader deployment in complex multi-agent scenarios.

\section{Related Work}
\textbf{In-context RL.} By framing RL as a sequence modeling problem, Decision Transformer (DT)~\cite{dt} can make decisions based on provided prompts~\cite{promptdt}. Subsequent studies scaled up model size and training data, enabling agents to exhibit ICL capabilities~\cite{mgdt,gato}. Algorithm Distillation~\cite{ad} takes a significant step towards ICRL by utilizing historical trajectories. This enables agents to automatically improve their performance through trial and error, without updating their parameters. Agentic Transformer (AT)~\cite{at} further demonstrates that cross-episodic contexts can help agents leverage hindsight, thus enabling performance improvement at test time~\cite{idt}. Decision-Pretrained Transformer (DPT)~\cite{dpt} explores an alternative approach by predicting the optimal action given random historical trajectories and the current state. Subsequently, Retrieval-Augmented Decision Transformer (RADT)~\cite{radt} introduces retrieval augmentation into ICRL, utilizing a DT-based embedding model to select relevant historical trajectories and thereby further aid action prediction. However, these methods have only demonstrated effectiveness on single-agent tasks with simple interactions~\cite{regent} and perform poorly on more complex decentralized cooperative tasks. In contrast, our approach achieves efficient trajectory retrieval tailored to the characteristics of Multi-Agent Systems (MASs), thereby facilitating collaborative adaptation to unseen tasks. To the best of our knowledge, our method is the first ICRL approach for Dec-POMDPs~\cite{decpomdp}.

\noindent\textbf{Cooperative multi-agent RL.} Many real-world problems are large-scale and complex, rendering single-agent modeling inefficient and often impractical~\cite{feng2025multi}. These challenges are more effectively addressed within the MAS setting~\cite{dorri2018multi}, where MARL provides a robust framework for solution~\cite{survey_marl}. In cooperative MARL, where agents pursue shared objectives, significant progress has been made in domains such as video games~\cite{li2025comprehensive}, domain calibration~\cite{jiang2024multi}, and financial trading~\cite{huang2024multi}. A central challenge in cooperative MARL is partial observability due to decentralized execution. The Centralized Training with Decentralized Execution (CTDE) framework~\cite{maddpg} addresses this by propagating team-level information to individual agents during training, thereby enhancing coordination at execution. Another key issue is the absence of individual rewards, which leads to the ``lazy agent" problem~\cite{vdn}, where agents fail to improve their policies due to an inability to assess their own contributions. Actor-critic methods such as COMA~\cite{coma} mitigate this by introducing counterfactual baselines for policy updates, while value-based approaches like QMIX~\cite{qmix} achieve implicit credit assignment by enforcing monotonicity in the value function~\cite{qtran,qplex}. In this work, we address both challenges within the ICRL framework by incorporating corresponding modules, thereby enabling rapid adaptation to unseen cooperative tasks.

\section{Background}
\subsection{Multi-Agent Reinforcement Learning} 
A cooperative multi-agent task is typically modeled as a Dec-POMDP~\cite{decpomdp}, defined by the tuple $\mathcal{M} = \langle \mathcal{S}, \mathcal{A}, \mathcal{T}, R, \Omega, O, \mathcal{N}, H, \rho \rangle$. Here, $\mathcal{S}$ and $\mathcal{A}$ denote the state and action spaces, respectively; $\rho \in \Delta(\mathcal{S})$ is the initial state distribution where $\Delta(\mathcal{S})$ represents the set of probability distributions over the state space $\mathcal{S}$; $H \in \mathbb{N}$ is the episode horizon where $\mathbb{N}$ denotes the set of natural numbers. Each episode begins with an initial state $s^0$ sampled from $\rho$. At each time step $h$, given the global state $s^h \in \mathcal{S}$, each agent $j \in \mathcal{N} = \{1, 2, \cdots, n\}$ receives a local observation $o_j^h \in \Omega$ generated by the observation function $O(s^h, j)$, and selects an action $a_j^h \in \mathcal{A}$ according to its individual learnable policy $\pi_j(a_j^h | \tau_j^h)$. Here, $\tau_j^h$ denotes the trajectory $(o_j^0,a_j^0,\cdots,o_j^h)$. The joint action is denoted as $\mathbf{a}^h = (a_1^h, a_2^h, \cdots, a_n^h)$. The environment then transitions to the next state $s^{h+1} \sim \mathcal{T}(\cdot | s^h, \mathbf{a}^h)$ and provides a global reward $r^h = R(s^h, \mathbf{a}^h)$. The episode terminates when a predefined condition is met or after $H$ steps. The objective is to optimize the joint policy $\bm{\pi} = (\pi_1, \pi_2, \cdots, \pi_n)$ to maximize the value function $V^{\mathcal{M}}(\bm{\pi}) = \mathbb{E}_{\bm{\pi}} \left[ \sum_{h=0}^{H-1} r^h \right]$.

\subsection{Decision Transformer}
Transformers, originally developed for sequence modeling in language tasks~\cite{vaswani2017attention}, have been applied to RL by Decision Transformer (DT)~\cite{dt}, which frames decision-making as sequence modeling~\cite{TT}. Instead of learning value functions as in traditional RL methods, DT derives policies from sequences of input tokens within a single trajectory, represented as $(\hat{R}^0, o^0, a^0, \hat{R}^1, o^1, a^1, \cdots)$, where each token corresponds to the Return-To-Go (RTG), observation, and action, respectively. The RTG at time step $h$, denoted as $\hat{R}^h$, is defined as the sum of future rewards: $\hat{R}^h = \sum_{t=h}^{H-1} r^t$. DT is trained in a supervised manner, similar to behavior cloning (BC)~\cite{BC}. During testing, by conditioning on a high RTG, DT can autoregressively generate actions aimed at achieving high cumulative rewards (return).

\subsection{Problem Setting} 
In this paper we study ICRL~\cite{survey_icrl}, a practical form of meta-RL~\cite{survey_metarl}, where agents learn new cooperative Dec-POMDP tasks sampled from $P(\mathcal{M})$ via limited online trials without updating model parameters. During training agents only access datasets $\mathcal{D}=\{\mathcal{D}_i\}$ of trajectories collected under unknown cooperative policies $\bm{\mu}$ in tasks $\mathcal{M}_i$. After pretraining, the model parameters are fixed. At test time, the agent team interacts with a new, unseen environment randomly sampled from $P(\mathcal{M})$ for only $T$ episodes, with the goal of achieving fast coordination without parameter updating. This objective can be formulated as maximizing the expected return in the final adaptation episode under the task distributions: $\max\mathbb{E}_{\mathcal{M}\sim P(\cdot)}V^{\mathcal{M}}(\bm{\pi})$. 

\section{Method}
\begin{figure*}[t]
\centering
\includegraphics[width=0.85\textwidth]{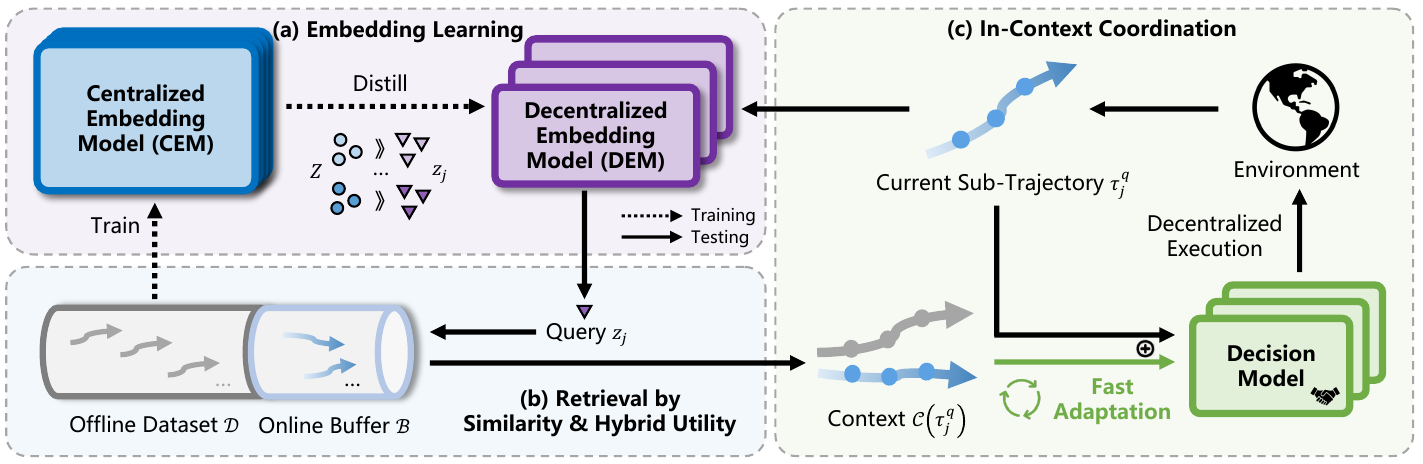} 
\caption{The conceptual workflow of MAICC. Dashed lines show data flow during centralized training, where CEM samples offline trajectories for training and distills team information to DEMs. Solid lines show data flow during decentralized execution, where sub-trajectories retrieve trajectories from the constructed memory based on similarity and hybrid utility score. Blue $\circ$ and purple $\triangledown$ denote different embeddings output by CEM and DEMs, respectively. $\oplus$ denotes concatenation of retrieved trajectories with the current sequence, which helps decision models adapt quickly.}
\label{fig:flow}
\end{figure*}
In this section, we present the MAICC (\textbf{M}ulti-\textbf{A}gent \textbf{I}n-\textbf{C}ontext \textbf{C}oordination) framework, which exploits the ICL capabilities of Transformer-based models for rapid adaptation to unseen cooperative tasks. The overall architecture is shown in Fig.~\ref{fig:flow}. During training, we first learn embedding models to capture the characteristics of multi-agent trajectories for efficient context retrieval (Sec.~\ref{method1}). Specifically, a centralized embedding model (CEM) extracts team-level information via autoregressive prediction, which guides decentralized embedding models (DEMs) for decentralized execution. The DEMs are then used to retrieve trajectories with similar embeddings for a given input, enabling the decision model to generate appropriate actions (Sec.~\ref{method2}). By leveraging these in-context trajectories, the pretrained decision model can infer task characteristics and generalize across diverse tasks. In the testing phase, we introduce a novel memory mechanism that combines an online replay buffer with offline datasets to enhance retrieval efficiency. We further propose a hybrid utility score that integrates individual- and team-level information to select high-quality in-context trajectories, promoting effective coordination (Sec.~\ref{method3}). Finally, we provide a theoretical analysis of the online cumulative regret of our approach (Sec.~\ref{method4}).

\subsection{Multi-Agent Trajectory Embedding Models}
\label{method1}
Efficient multi-agent trajectory retrieval relies on learning high-quality trajectory embeddings. To achieve this, we adopt the CTDE paradigm~\cite{ctde2016,maddpg} and design both centralized and decentralized embedding models. During training, agents have access to global team observations and actions, allowing the CEM to capture fine-grained team-level information. In contrast, during execution, each agent is limited to its own local observations and actions, resulting in less informative embeddings from the DEMs. To address this disparity, we employ the CEM to distill team-level knowledge into the DEMs during training, thereby enhancing the DEMs' representational capacity for decentralized execution.

Formally speaking, we denote the number of agents as $n$. The multi-task offline dataset $\mathcal{D}$ consists of trajectories $\tau = (\bm{o}^0, \bm{a}^0, r^0, \ldots, \bm{o}^{H-1}, \bm{a}^{H-1}, r^{H-1})$, where $\bm{o} = (o_1, \ldots, o_n)$ and $\bm{a} = (a_1, \ldots, a_n)$. Our trajectory embedding models employ three types of tokens: observation $o$, action $a$, and post-step information $\hat{P}$. Following prior work~\cite{at,idt}, the token $\hat{P}$ comprises the global reward, done signal, and task completion flag, which are essential for modeling long-horizon trajectories. We omit the RTG token, as it can cause retrieval of trajectories from irrelevant tasks that happen to have similar RTG values, thereby reducing the informativeness of in-context examples and harming action prediction. 

As illustrated in Fig.~\ref{fig:cem}, the CEM receives the agents’ local observations $\{o^h_j\}_{j=1}^n$, actions $\{a^h_j\}_{j=1}^n$, and post-step information $\hat{P}^h$ at each time step $h$, and outputs the corresponding embeddings: $\{Z^h_{o,j}\}_{j=1}^n, \{Z^h_{a,j}\}_{j=1}^n, Z_p^h = \mathrm{CEM}(\{o^h_j\}_{j=1}^n, \{a^h_j\}_{j=1}^n, \hat{P}^h)$. To be compatible with centralized training, we adapt the causal transformer by introducing intra-team visibility, enabling observation and action tokens within the same team and time step to attend to each other. We further design three loss functions to model the behavior policy ($\mathcal{L}_{\mu}$), reward function ($\mathcal{L}_R$), and observation transition dynamics ($\mathcal{L}_{\mathcal{T}}$) of the trajectory:
\begin{align}
&\mathcal{L}_{\mathrm{CEM}}=\mathcal{L}_{\mu}+\mathcal{L}_R+\mathcal{L}_{\mathcal{T}},\label{eqn.cem}\\
&\mathcal{L}_{\mu}=-\mathbb{E}_{\tau\sim\mathcal{D}}\sum_{h=0}^{H-1}\sum_{j=1}^n \log \mathrm{MLP}_{o\to a}(a^h_j|Z_{o,j}^h),\label{loss_mu}\\
&\mathcal{L}_R=\mathbb{E}_{\tau\sim\mathcal{D}}\sum_{h=0}^{H-1}\left(r^h-\sum_{j=1}^n \mathrm{MLP}_{a\to r}(Z_{a,j}^h)\right)^2,\label{loss_r}\\
&\mathcal{L}_{\mathcal{T}}=-\mathbb{E}_{\tau\sim\mathcal{D}}\sum_{h=0}^{H-2}\sum_{j=1}^n\log\mathrm{MLP}_{p\to o}(o^{h+1}_j|Z_p^h,o^h_j),
\end{align}
where $\mathrm{MLP}$s with different subscripts fit different functions. Eq.~\ref{loss_r} can be regarded as performing implicit credit assignment~\cite{vdn,qmix}, which benefits subsequent decentralized adaptation. 

\begin{figure}[t]
\centering
\includegraphics[width=0.82\columnwidth]{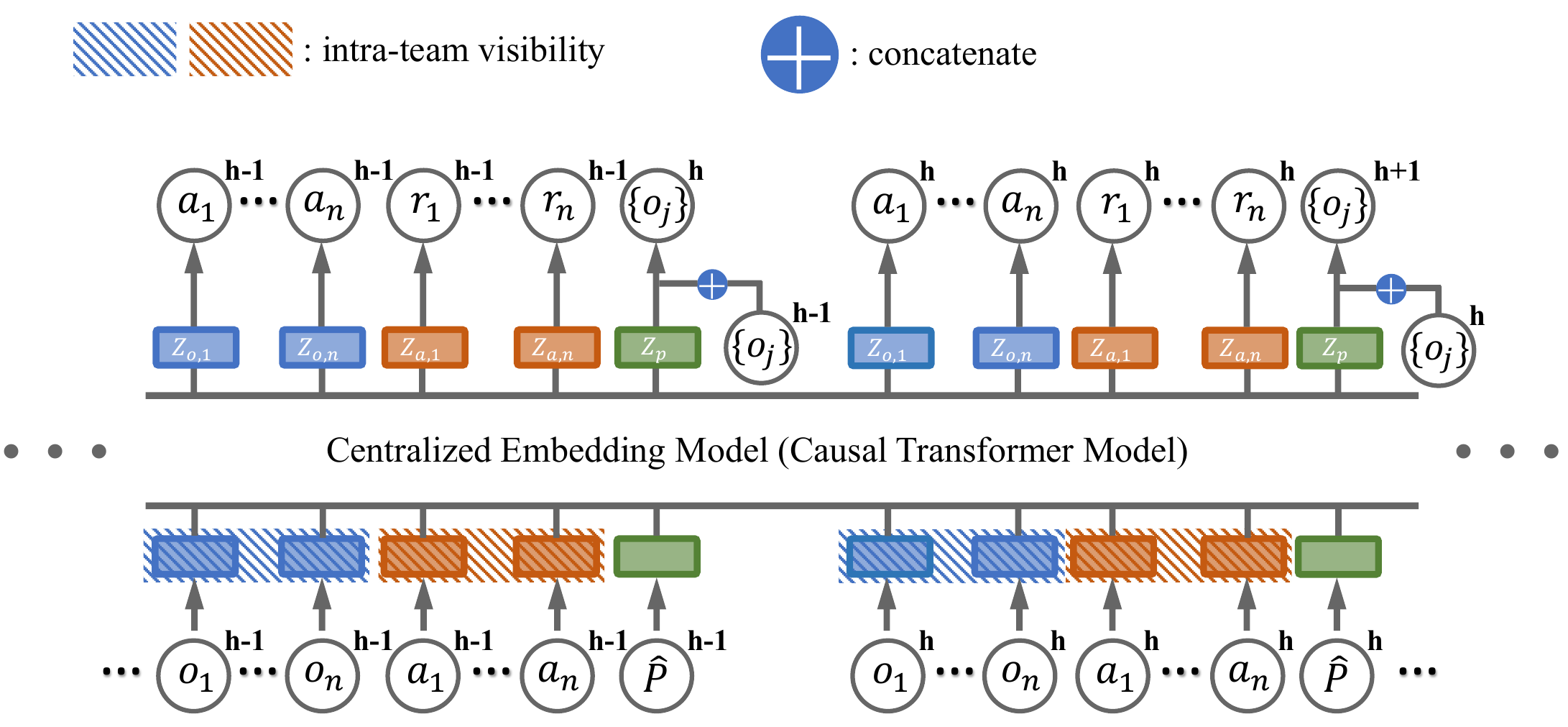} 
\caption{The illustration of CEM. Intra-team visibility enables observation and action tokens within the same team to attend to each other at each time step. The causal transformer predicts individual actions and rewards, while the post-step information token, concatenated with the previous individual observation, is used to predict the next observation.}
\label{fig:cem}
\end{figure}
During decentralized execution, the DEMs capture embeddings using only local information, i.e., $z^h_{o,j}, z^h_{a,j}, z_p^h = \mathrm{DEM}(o^h_j, a^h_j, \hat{P}^h)$. To enhance coordination, we introduce auxiliary objectives that distill team-level information by minimizing the KL divergence between the embeddings generated by the CEM and those produced by the DEMs:
\begin{align}
\mathcal{L}_{\mathrm{DEM}}&=\mathbb{E}_{\tau\sim\mathcal{D}}\sum_{h=0}^{H-1}\sum_{j=1}^n\left[\mathrm{KL}(Z^h_{o,j},z^h_{o,j})+\mathrm{KL}(Z^h_{a,j},z^h_{a,j})\right]\notag\\
    &+\mathbb{E}_{\tau\sim\mathcal{D}}\sum_{h=0}^{H-1}\mathrm{KL}(Z^h_{p},z^h_{p})\label{eqn.dem},
\end{align}
where $\mathrm{KL}(p,q)$ measures the divergence from the target distribution $p$ to the approximate distribution $q$.

\subsection{Retrieval-Based In-Context Decision Training}
\label{method2}
To address diverse cooperative tasks with a single decision model, we use the trained DEMs to retrieve trajectories that inform action generation. Given an individual query sub-trajectory $\tau^q_j = (o^0_j, a^0_j, r^0, \cdots, a^{q-1}_j, r^{q-1}, o^q_j) \sim \mathcal{D}$ up to a certain time step, we first input it into the DEM and extract the embeddings at the final step, which, due to the transformer's long-range dependency modeling, summarizes the entire sub-trajectory. We then apply average pooling on the extracted embeddings over different tokens to obtain the final query embedding, i.e., $z^q_j = \mathrm{MEAN}(z^q_{o,j}, z^{q-1}_{a,j}, z^{q-1}_p)$. Using Maximum Inner Product Search (MIPS)~\cite{faiss}, we retrieve the top-$k$ most relevant in-context trajectories: $\mathcal{C}(\tau^q_j) = \arg\max^k_{\tau^c \in \mathcal{D}} \mathrm{cossim}(z^c, z^q_j)$, where $k$ is the number of in-context trajectories, $\mathrm{cossim}$ denotes cosine similarity, and $z^c$ is the candidate embedding computed in the same manner as the query.

The retrieved in-context trajectories provide additional task-specific information to the current sub-trajectory. We concatenate these trajectories with the query and train the decision model $\pi_\theta$ (a causal transformer with parameter $\theta$ sharing across agents) using the following loss function:
\begin{align}
    \mathcal{L}_\pi=-\mathbb{E}_{\tau^q_j\sim\mathcal{D}}\log\pi_\theta\left(a^q_j|\mathrm{CONCAT}(\mathcal{C}(\tau^q_j),\tau^q_j)\right), \label{eqn.pi}
\end{align}
where $\mathrm{CONCAT}$ denotes the concatenate function. It is worth noting that, in addition to the three types of tokens—observation, action, and post-step information—the decision model also receives a RTG token. Unlike the embedding models, the decision model leverages the RTG from the retrieved trajectory to guide action selection towards achieving the desired return.

\subsection{Decentralized In-Context Fast Coordination}
\label{method3}
After pretraining the embedding and decision models, a new task is randomly sampled from the task distribution $P(\mathcal{M})$. The agent team must then rapidly adapt and coordinate on this task without further parameter updates. During $T$ episodes of interaction, data are stored in an online replay buffer. Agents can retrieve trajectories from both the multi-task offline dataset $\mathcal{D}$, which may exhibit distribution shift, and the online buffer $\mathcal{B}$, which is aligned with the current task but initially contains limited experience. To address this, we propose a selective memory mechanism with exponential time decay: early episodes prioritize offline data to encourage exploration, while later episodes increasingly leverage high-value online trajectories to enhance exploitation. Specifically, we introduce a coefficient $\beta_t = \exp\left(-\lambda \frac{t}{T}\right)$ for episode $t$, where the hyper-parameter $\lambda$ controls the decay rate~\cite{dagger}. We construct a new buffer $\mathcal{B}'$ by sampling from $\mathcal{D}$ with probability $\beta_t$ and from $\mathcal{B}$ with probability $1-\beta_t$. This method is simple, effective, and theoretically grounded.

Based on the constructed memory $\mathcal{B}'$, we further enhance the exploitation of high-value trajectories by introducing a hybrid utility score during inference, defined as $\mathcal{S}_{\mathrm{util}}(\tau) = \alpha\, \mathrm{norm}(\mathcal{R}) + (1-\alpha)\, \mathrm{norm}(\tilde{\mathcal{R}})$. Here, $\mathcal{R} = \sum_{h=0}^{H-1} r^h$ is the global return, $\tilde{\mathcal{R}} = \sum_{h=0}^{H-1} \tilde{r}^h_j$ is the predicted individual return for agent $j$, $\mathrm{norm}(\cdot)$ denotes normalization to $[0, 1]$, and $\alpha \in [0, 1]$ is a hyper-parameter. In Dec-POMDPs, where individual rewards are unavailable, we leverage the pretrained DEMs to predict individual rewards from action embeddings, i.e., $\tilde{r}^h_j = \mathrm{MLP}_{a \to r}(z^h_{a,j})$. This hybrid utility score enables agents to retrieve trajectories that are beneficial at both the individual and team levels, thereby mitigating the ``lazy agent" problem in multi-agent systems. Incorporating the similarity score used during training, the retrieval process is formulated as $\mathcal{C}(\tau^q_j) = \arg\max^k_{\tau^c} \mathcal{S}(\tau^c, \tau^q_j)$, where $\tau^c \in \mathcal{B}'$ and $\mathcal{S}(\tau^c, \tau^q_j) = \mathrm{cossim}(z^c, z^q_j) + \mathcal{S}_{\rm{util}}(\tau^c)$. The decision model then outputs actions conditioned on the concatenation of the retrieved in-context trajectories and the input trajectory: $a \sim \pi_\theta(\cdot\,|\,\mathrm{CONCAT}(\mathcal{C}(\tau^q_j), \tau^q_j))$. The overall pseudo code of MAICC is provided in Alg.~\ref{algorithm}.

\begin{algorithm}
\caption{Multi-agent In-context Coordination via Decentralized Memory Retrieval}
\label{algorithm}
\textbf{Input}: Initialized two trajectory embedding models $\mathrm{CEM}$, $\mathrm{DEM}$, decision model $\pi_\theta$, multi-task offline dataset $\mathcal{D}$, empty online replay buffer $\mathcal{B}$
\begin{algorithmic}[1]
\STATE // Multi-Agent Trajectory Embedding Models
\WHILE{not converged}
\STATE Update $\mathrm{CEM}$ and $\mathrm{DEM}$ by Eq.~\ref{eqn.cem} and Eq.~\ref{eqn.dem} on $\mathcal{D}$
\ENDWHILE
\STATE // Retrieval-Based In-Context Decision Training
\WHILE{not converged}
\STATE Retrieve in-context trajectories $\mathcal{C}$ with $\mathrm{DEM}$
\STATE Update $\pi_\theta$ with $\mathcal{C}$ by Eq.~\ref{eqn.pi}
\ENDWHILE
\STATE // Decentralized In-Context Fast Adaptation
\FOR{$t=1,2,\cdots,T$}
\STATE Construct the memory $\mathcal{B}'$
\WHILE{episode not ended}
\STATE Retrieve in-context trajectories $\mathcal{C}$ with $\mathcal{S}$ and $\mathcal{B}'$
\STATE Decentralized execution with $\pi_\theta$ conditioned on $\mathcal{C}$
\ENDWHILE
\STATE  Store episode trajectory in $\mathcal{B}$
\ENDFOR
\end{algorithmic}
\end{algorithm}
\subsection{Theoretical Analysis}
\label{method4}
In this section, we provide a bound on the online cumulative regret of MAICC. For a given task $\mathcal{M}$ with $|\Omega| = \omega$, $|\mathcal{A}| = A$, and horizon $H$, let $\pi^*$ denote the expert policy. The cumulative regret over $T$ episodes is defined as $\mathbf{Reg}_\mathcal{M} = \sum_{t=1}^T V^\mathcal{M}(\pi^*) - V^\mathcal{M}(\hat{\pi}_t)$, where $\hat{\pi}_t = \beta_t \pi^{\mathcal{D}} + (1-\beta_t) \pi^{\mathcal{B}}_t$ with subscript $t$ is the MAICC policy in episode $t$. Here, $\pi^{\mathcal{D}}$ is the policy that retrieves the in-context trajectories from the offline dataset $\mathcal{D}$, while $\pi^{\mathcal{B}}_t$ retrieves them from the online buffer $\mathcal{B}$ accumulated up to episode $t$, described in Sec.~\ref{method3}.

\begin{assumption}
\textbf{(Sufficiency of Retrieval)} Let $\pi^{\mathcal{B}*}_t$ denote the policy that, for each query $\tau^q$, directly uses the entire online buffer accumulated over $t$ episodes as the in-context input (i.e., without retrieval). For all $(\tau^q, \mathcal{B}, t)$, we have $\pi^\mathcal{B}_t(a|\tau^q) = \pi^{\mathcal{B}*}_t(a|\tau^q)$ for all $a \in \mathcal{A}$.
\end{assumption}

This assumption is trivially satisfied if the number of retrieved trajectories $k$ equals $t$. Even when $k < t$, a carefully selected $k$ trajectories can still capture most of the relevant information. Since Transformer inference time scales quadratically with context length, using a representative subset rather than the entire buffer is both efficient and practical.

\begin{theorem}
Suppose $\sup_\mathcal{M} P(\mathcal{M})/P_{\mathcal{D}}(\mathcal{M}) \le C$ for some $C > 0$, where $P_{\mathcal{D}}(\mathcal{M})$ denotes the training task distribution. Then the expected online cumulative regret of MAICC satisfies $\mathbb{E}_{P(\mathcal{M})}[\mathbf{Reg}_\mathcal{M}] \le \tilde{\mathcal{O}}(CH^{3/2}\omega\sqrt{AT})$.
\end{theorem}

MAICC offers a theoretical guarantee similar to prior ICRL methods~\cite{radt,dpt,tao} as $\tilde{\mathcal{O}}$ is Big-O ignoring poly-logarithmic factors in complexity. In practice, however, the initial online replay buffer may lack sufficiently informative trajectories, leading to inefficient exploration. By leveraging our selective memory mechanism, MAICC adapts to new tasks more efficiently. Experimental results further support this advantage, and detailed derivations are provided in Appendix~\ref{App.B}.

\section{Experiments}
In this section, we evaluate the proposed MAICC framework empirically. We begin by describing the experimental environments and baseline methods in Sec.~\ref{sec:exp_setup}. We then conduct a series of experiments to address the following questions: (1) How does MAICC compare to various baselines in terms of fast coordination (Sec.~\ref{sec:performance})? (2) How effectively do the DEMs capture representations of multi-agent trajectories (Sec.~\ref{sec:visualization})? (3) What is the contribution of each component of MAICC to overall performance (Sec.~\ref{sec:ablation_study})?

\subsection{Experiment Setup}
\label{sec:exp_setup}
\begin{figure*}[t]
\centering
\includegraphics[width=0.98\textwidth]{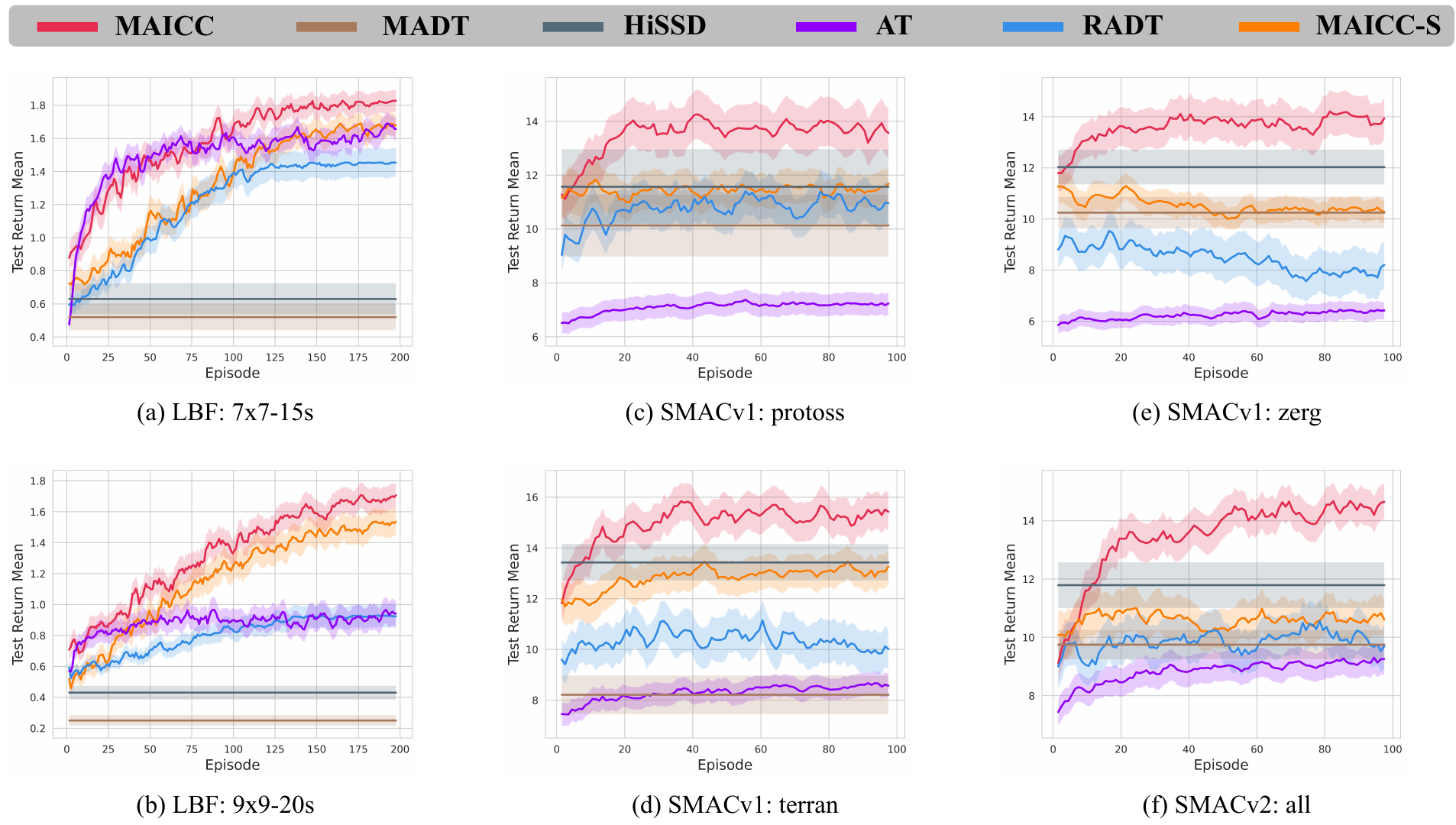} 
\caption{In-context adaptation performance across different scenarios. Each scenario is evaluated over 50 test runs on randomly sampled tasks, with results reported as the mean return and 95\% confidence interval.}
\label{fig:main_results}
\end{figure*}
We evaluate MAICC and baseline methods on several cooperative benchmarks. The first is the \textbf{Level-Based Foraging (LBF)}~\cite{lbf}, a grid-world environment where agents must coordinate to collect food items simultaneously. Each agent observes only its local field of view and must collect food at different locations for each task, within a limited number of time steps. We consider two scenarios: \textit{LBF: 7x7-15s} and \textit{LBF: 9x9-20s}, which differ in grid size and time limits. We further assess MAICC on the \textbf{StarCraft Multi-Agent Challenge (SMAC)}~\cite{smac}, using three sets of tasks where \textit{Protoss}, \textit{Terran}, and \textit{Zerg} units cooperate to defeat enemy units controlled by the built-in AI of the same race. Each task features varying agent types and numbers, with corresponding enemy configurations. Additionally, we evaluate on the \textbf{StarCraft Multi-Agent Challenge-v2 (SMACv2)}~\cite{smacv2}, an extension of SMAC with increased randomness. For this benchmark, we further challenge MAICC by pretraining a single model to handle \textit{all} three task types. For each scenario, we use QMIX~\cite{qmix} to train on multiple tasks, forming the multi-task offline dataset $\mathcal{D}$. Further details on these benchmarks and datasets are provided in Appendix~\ref{App.c}.

MAICC is pretrained on a multi-task offline dataset and learns rapid coordination through online decentralized adaptation. For comparison, we select several baselines with similar settings. MADT~\cite{madt} extends DT to the multi-agent domain and achieves strong performance in single-task scenarios. AT~\cite{at} and RADT~\cite{radt} are state-of-the-art in-context RL algorithms trained on offline data for online adaptation; while effective in single-agent settings, they lack designs specific to multi-agent coordination. HiSSD~\cite{hissd} is a recent multi-task MARL algorithm that learns generalizable skills from a multi-task offline dataset, but does not support online adaptation. MAICC-S is an ablated version of our method, where only the DEM is trained for trajectory modeling during pretraining, without the CEM; all other components remain unchanged. Except for HiSSD, all methods are Transformer-based, and we use the same-size GPT-2 model~\cite{gpt2} for fair comparison.

Experimental results are obtained by training each model with 5 different random seeds. For each seed, performance is evaluated on 10 random tasks, yielding a total of 50 test runs. We report the mean and 95\% confidence intervals. Additional implementation details are provided in Appendix~\ref{App.D}.

\subsection{Main Results}
\label{sec:performance}
\begin{figure}
\centering
\includegraphics[width=0.95\columnwidth]{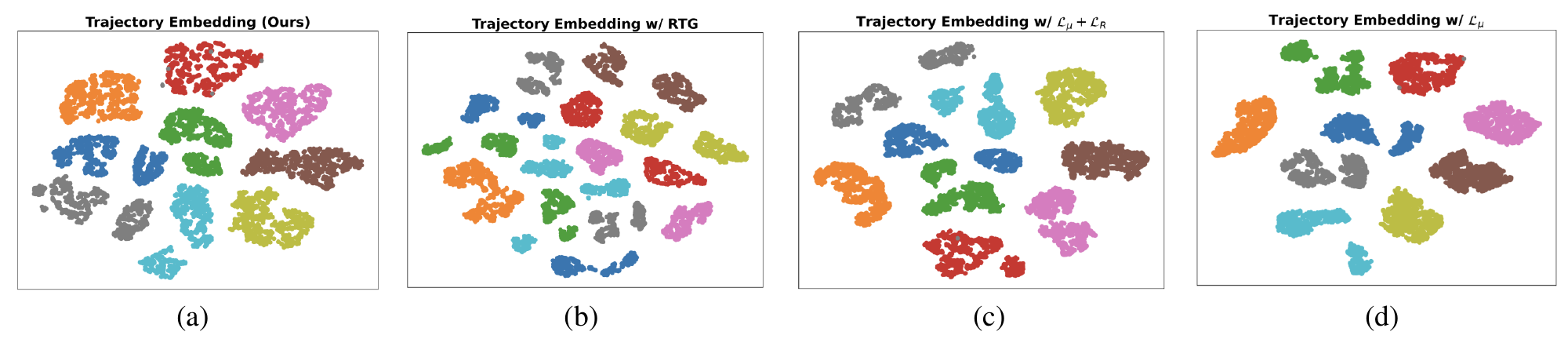} 
\caption{Visualization results illustrating the effects of different embedding model training settings. Each point in the figure represents the embedding of a trajectory from the dataset, with points of the same color corresponding to trajectories from the same task.}
\label{fig:vis}
\end{figure}
We first evaluate the in-context adaptation performance of our method and the baselines across various scenarios. As shown in Fig.~\ref{fig:main_results}, agent teams are required to improve their average return over the task distribution within a limited number of episodes in each scenario. Our method consistently outperforms all baselines across six test scenarios, achieving faster adaptation to unseen cooperative tasks without requiring model parameter updates.

Since MADT and HiSSD lack online adaptation capabilities, their performance is shown as fixed horizontal lines. On SMAC-type tasks, their results are comparable to in-context RL baselines; however, on LBF tasks—where agent observability is more limited—their performance drops significantly, underscoring the importance of in-context adaptation. AT predicts actions based on trajectories from previous episodes, yielding good results only on the small $\textit{LBF: 7x7-15s}$ map. Although RA-DT also utilizes trajectory retrieval, its coarse-grained encoding and lack of adaptation for cooperative scenarios limit its effectiveness. The performance gap between MAICC-S and our method further demonstrates the necessity of explicitly modeling multi-agent characteristics in trajectory embeddings. Notably, in more complex SMAC scenarios, only our method exhibits clear in-context adaptation. The performance gap is most pronounced in the \textit{SMACv2: all} scenario, which features the greatest task diversity, highlighting the strong potential of our approach in large-scale data settings.
\subsection{Visualization of Learned Embeddings}
\label{sec:visualization}
To assess the effectiveness of the trajectory embedding model, we conduct visualization experiments. As shown in Fig.~\ref{fig:vis}, for the \textit{SMACv2: all} scenario, all trajectories are encoded and projected onto a two-dimensional plane using t-SNE~\cite{tsne}. Points with the same color represent trajectories from the same task.

We evaluate four different embedding configurations. In our proposed setting (Fig.~\ref{fig:vis}(a)), the embedding models are trained without the RTG token and utilizes three loss functions, resulting in fine-grained embeddings where trajectories from the same task are grouped. In contrast, incorporating the RTG token (Fig.~\ref{fig:vis}(b)) causes embeddings from the same task to form several small, overlapping clusters with those from other tasks, increasing the risk of retrieving irrelevant trajectories. In Fig.~\ref{fig:vis}(c) and (d), only a subset of the loss functions is used to model trajectories~\cite{radt}. While the model still encodes trajectories from the same task into nearby regions, overfitting occurs, resulting in overly compact clusters due to coarse-grained modeling. When tested on unknown tasks, such trajectory representations lack the generalization capability for extrapolated estimation. These findings highlight the importance of carefully designing both the embedding models and its associated loss functions for effective trajectory retrieval.

\subsection{Ablation Study}
\label{sec:ablation_study}
\begin{table*}[t]
  \centering
  \begin{tabular}{c|cccc|c}
    \toprule
    \textbf{Variants} & \multicolumn{1}{c}{\textbf{EM With RTG}} & \multicolumn{1}{c}{\textbf{Coefficient $\beta$}} & \multicolumn{1}{c}{\textbf{CEM loss}} & \multicolumn{1}{c|}{\textbf{Hyper-parameter $\alpha$}} &  \multicolumn{1}{c}{\textbf{\textit{SMACv2: all} Ret.}} \\
    \midrule
    \textbf{Default} & False & $\beta_t=\exp(-\lambda \frac{t}{T})$ & $\mathcal{L}_{\mu}+\mathcal{L}_R+\mathcal{L}_{\mathcal{T}}$ & $\alpha=0.8$ & 14.51$\pm$0.46\\ 
    \midrule
    (A)& True & & & & 13.52$\pm$0.62 \\
    \midrule
    \multirow{2}{*}{(B)}&  &  $\beta_t=0$ & & & 12.16$\pm$0.72\\
    & & $\beta_t=1$& & & 11.17$\pm$0.64\\
    \midrule
    \multirow{3}{*}{(C)}&  &  & $\mathcal{L}_{\mu}+\mathcal{L}_R$ & &13.43$\pm$0.51\\
    & & & $\mathcal{L}_{\mu}+\mathcal{L}_{\mathcal{T}}$ & $\alpha=1$ &12.32$\pm$0.48\\
    & & & $\mathcal{L}_{\mu}$ & $\alpha=1$ &10.55$\pm$0.39\\
    \midrule
    \multirow{2}{*}{(D)}& & & & $\alpha=1$ & 13.61$\pm$0.40\\
    & & & & $\alpha=0$ & 13.26$\pm$0.66\\
    \bottomrule
  \end{tabular}
  \caption{Ablation Study on MAICC. Unless otherwise noted, all settings follow the default configuration. “Ret.” indicates the average return over 50 test runs (with 95\% confidence interval), evaluated in the final adaptation episode.}
\label{table_abl}
\end{table*}
We evaluated the importance of different MAICC components by systematically modifying the default model and measuring performance changes on the most challenging scenario, \textit{SMACv2: all}, as shown in Tab.~\ref{table_abl}.

In row (A), we examine the effect of incorporating the RTG token during embedding model training. The results show degraded performance, likely due to an increased likelihood of retrieving irrelevant trajectories.

Row (B) explores different values of $\beta$ for memory construction. When the memory consists solely of either the offline dataset or the online buffer—instead of combining both sources using exponential time decay as coefficient—performance drops significantly. This indicates that each data source has limitations, and their weighted combination is crucial for effective adaptation to unseen tasks.

In row (C), we analyze the impact of different CEM loss functions on in-context adaptation. The results indicate that all three loss functions are necessary; fine-grained trajectory modeling enhances action prediction. Notably, omitting $\mathcal{L}_R$ prevents individual return prediction during testing, further reducing overall performance.

Row (D) evaluates the role of the hybrid utility score. Using only the global return $(\alpha=1)$ leads to insufficient credit assignment, while relying solely on the predicted individual return $(\alpha=0)$ may suffer from prediction inaccuracies. Therefore, the hybrid approach, which combines both, yields improved adaptation performance.

\section{Conclusion and Discussion}
In this paper, we address rapid cooperative adaptation in Dec-POMDP settings by proposing the MAICC framework, which enables agent teams to quickly coordinate on unseen tasks without requiring parameter updates. During training, MAICC leverages the CEM to extract fine-grained representations of multi-agent trajectories and guides the DEM to optimize these representations for decentralized execution. Given a current sub-trajectory, agents use the DEM to retrieve and concatenate relevant trajectories for decision model training. During testing, each agent retrieves trajectories from a constructed memory that integrates both online buffer and offline data. Credit assignment is achieved by combining team- and individual-level returns. Experiments on cooperative MARL benchmarks demonstrate that MAICC enables rapid adaptation to previously unseen tasks. A potential constraint of MAICC is that relying solely on exponential time decay for memory construction may limit applicability in certain scenarios; incorporating uncertainty-based metrics~\cite{surver_unc} could further enhance generalization and facilitate real-world deployment.

\clearpage
\newpage

\bibliographystyle{ACM-Reference-Format}
\bibliography{aaai24}

\newpage
\onecolumn
\appendix

\setcounter{secnumdepth}{2} 
\section{Extended Related Work}
\noindent\textbf{Offline meta-RL.} Offline Reinforcement Learning (RL) aims to learn policies from static datasets collected by a behavior policy, without further interaction with the environment~\cite{levineofflinesurvey}. Offline Meta-RL (OMRL) extends this paradigm by training on a distribution of such offline tasks, thereby enabling generalization to novel tasks. Among existing approaches, context-based meta-learning methods employ a context encoder to perform approximate inference over task representations, conditioning the meta-policy on the inferred belief to improve generalization. For example, FOCAL~\cite{focal} introduces a distance metric learning loss to distinguish between tasks, while subsequent works~\cite{corro,gentle,zhang2024debiased} focus on accurately identifying environment dynamics and debiasing representations from the behavior policy. However, this paradigm can be unstable due to function approximation and bootstrapping, particularly in the absence of further online adaptation~\cite{metadt}. In contrast, gradient-based OMRL methods adapt to unseen tasks with a few interactions~\cite{linmodel}, but the required online gradient updates can be prohibitively expensive in real-world applications. Recently, in-context reinforcement learning (ICRL)~\cite{survey_icrl} has emerged as an alternative, enabling rapid adaptation through trial-and-error without parameter updates. This paradigm demonstrates significant potential for practical applications.

\noindent\textbf{Online RL with offline datasets.} Online RL is often impractical in high-risk domains such as autonomous driving~\cite{survey_driving}, due to the high cost and safety concerns associated with trial-and-error learning. In contrast, offline RL leverages large amounts of pre-collected data to safely update policies, but suffers from the distribution shift problem~\cite{levineofflinesurvey}. To address the limitations of both approaches, online RL with offline datasets has emerged as a promising solution~\cite{offline2online}. One line of work pre-trains policies using offline RL and subsequently performs online fine-tuning~\cite{hester2018deep,lee2022offline}. However, the mismatch between offline data and online samples often results in inefficient adaptation during fine-tuning. Another line of research encourages the online agent to mimic behaviors found in the offline data~\cite{levine2013guided,awac}. While this can stabilize policy updates, it still requires a large number of online samples to accomplish tasks due to the inherent limitations of the offline data. A third line of work initializes the replay buffer with offline data, which has been theoretically shown to provide strong guarantees and achieve excellent performance in both theory and practice~\cite{nair2018overcoming,songhybrid}. Our method is similar to this third category: we design a retrieval mechanism based on exponential time decay over the replay buffer and offline data, enabling the agent to quickly adapt to deployment environments.

\section{Detailed Theoretical Analysis}
\label{App.B}
We first show that, under Assumption 1, $\pi_t^\mathcal{B}$ is equivalent to the online test setting in DPT~\cite{dpt}. In DPT, the model predicts optimal actions during both training and testing by leveraging an in-context dataset sampled from the same task, together with the current context. When Assumption 1 holds, the in-context trajectories retrieved by MAICC’s embedding model can serve the same purpose as the in-context dataset in DPT. Therefore, under these conditions, the two approaches are equivalent. To facilitate the derivation of the online cumulative regret for MAICC, we first present a lemma that bounds the total variation distance between the distributions of observations encountered by $\hat{\pi}_t$ and $\pi^\mathcal{B}_t$~\cite{dagger}, where $\hat{\pi}_t = \beta_t \pi^\mathcal{D} + (1-\beta_t)\pi^\mathcal{B}_t$.

\begin{lemma}
    $||d_{\hat{\pi}_t} - d_{\pi^\mathcal{B}_t}|| \le 2H\beta_t$.
\end{lemma}

\begin{proof}
    Let $d$ denote the distribution of observations over $H$ steps conditioned on $\hat{\pi}_t$ selecting $\pi^\mathcal{D}$ at least once during these $H$ steps. Since $\hat{\pi}_t$ executes $\pi^\mathcal{B}_t$ exclusively over $H$ steps with probability $(1-\beta_t)^H$, we have $d_{\hat{\pi}_t} = (1-\beta_t)^H d_{\pi^\mathcal{B}_t} + (1-(1-\beta_t)^H)d$. Thus,
    \begin{align}
        ||d_{\hat{\pi}_t} - d_{\pi^\mathcal{B}_t}||_1 &= (1-(1-\beta_t)^H) ||d - d_{\pi^\mathcal{B}_t}||_1 \notag \\
        &\le 2(1-(1-\beta_t)^H) \notag \\
        &\le 2H\beta_t,
    \end{align}
    where the last inequality follows from the fact that $(1-\beta)^H \ge 1 - \beta H$ for any $\beta \in [0,1]$.
\end{proof}

Next, if we set $\beta_t$ as an exponential decay over time $t$, we can obtain: \begin{theorem}
Suppose $\sup_\mathcal{M} P(\mathcal{M})/P_{\mathcal{D}}(\mathcal{M}) \le C$ for some $C > 0$, where $P_{\mathcal{D}}(\mathcal{M})$ denotes the training task distribution. Then the expected online cumulative regret of MAICC satisfies $\mathbb{E}_{P(\mathcal{M})}[\mathbf{Reg}_\mathcal{M}] \le \tilde{\mathcal{O}}(CH^{3/2}\omega\sqrt{AT})$.
\end{theorem}
\begin{proof}
For a given task $\mathcal{M}\in P_{\mathcal{D}}(\cdot)$, when the reward is bounded within $[0,1]$, the cumulative performance gap between $\hat{\pi}_t$ and $\pi^\mathcal{B}_t$ is bounded as follows:
\begin{align*}
    \sum_{t=1}^T  V^{\mathcal{M}}(\pi^\mathcal{B}_t) - V^{\mathcal{M}}(\hat{\pi}_t)  &\leq \sum_{t=1}^T \sum_{h=1}^H 2h\beta_t\\
    &\le\sum_{t=1}^T2H^2\beta_t.
\end{align*}
Although this appears to be $\mathcal{O}(H^2)$, if we choose $\beta_t = \gamma^{t}$ with $\gamma \in [0,1]$, then there exists an $n_\beta$, defined as the largest $n \leq T$ such that $\beta_n > 1/H$. In this case, the bound can be further refined as follows:
\begin{align*}
    \sum_{t=1}^T  V^{\mathcal{M}}(\pi^\mathcal{B}_t) - V^{\mathcal{M}}(\hat{\pi}_t) &\le 2H\sum_{t=1}^T\min(1,H\beta_t)\\
    &= 2H(n_\beta+H\sum_{t=n_\beta+1}^T\beta_t)\\
    &\le 2H\frac{\log H+1}{1-\gamma},
\end{align*}
which becomes $\mathcal{O}(H\log H)$. In addition, $\beta_t = \gamma^{t}$ is equivalent to $\beta_t = \exp(-\lambda \frac{t}{T})$ as used in the main paper, if we set $\lambda = -T \log \gamma$. Following the results of previous work~\cite{dpt,osband2013more}, we can obtain that the cumulative regret between $\pi^\mathcal{B}_t$ and the expert policy $\pi^*$ is given by:
\begin{align*}
    \sum_{t=1}^T V^{\mathcal{M}}(\pi^*) - V^{\mathcal{M}}(\pi^\mathcal{B}_t) \le \tilde{\mathcal{O}}(H^{3/2}\omega\sqrt{AT}),
\end{align*}
therefore the cumulative regret between $\hat{\pi}_t$ and the expert policy $\pi^*$ is bounded as:
\begin{align*}
    \mathbf{Reg}_\mathcal{M}&=\sum_{t=1}^T V^{\mathcal{M}}(\pi^*) - V^{\mathcal{M}}(\hat{\pi}_t) \\
    &\le \tilde{\mathcal{O}}(H^{3/2}\omega\sqrt{AT})+\mathcal{O}(H\log H)\\
    &=\tilde{\mathcal{O}}(H^{3/2}\omega\sqrt{AT}),
\end{align*}
when $T$ is huge enough. Since $\sup_\mathcal{M} P(\mathcal{M})/P_{\mathcal{D}}(\mathcal{M}) \le C$ for some $C > 0$, then 
\begin{align*}
    \mathbb{E}_{P(\mathcal{M})}[\mathbf{Reg}_\mathcal{M}]&=\int P(\mathcal{M})\mathbf{Reg}_\mathcal{M}d\mathcal{M}\\
    &\le C\int P_{\mathcal{D}}(\mathcal{M})\mathbf{Reg}_\mathcal{M}d\mathcal{M}\\
    &\le\tilde{\mathcal{O}}(CH^{3/2}\omega\sqrt{AT}).
\end{align*}
\end{proof}
In practice, since the online replay buffer may lack sufficiently informative trajectories in the early stages, our constructed memory performs better, as demonstrated by the experimental results.

\section{Extended Benchmark Descriptions}
\label{App.c}
\subsection{Level-Based Foraging}
\begin{figure}[t]
\centering
\includegraphics[width=0.7\columnwidth]{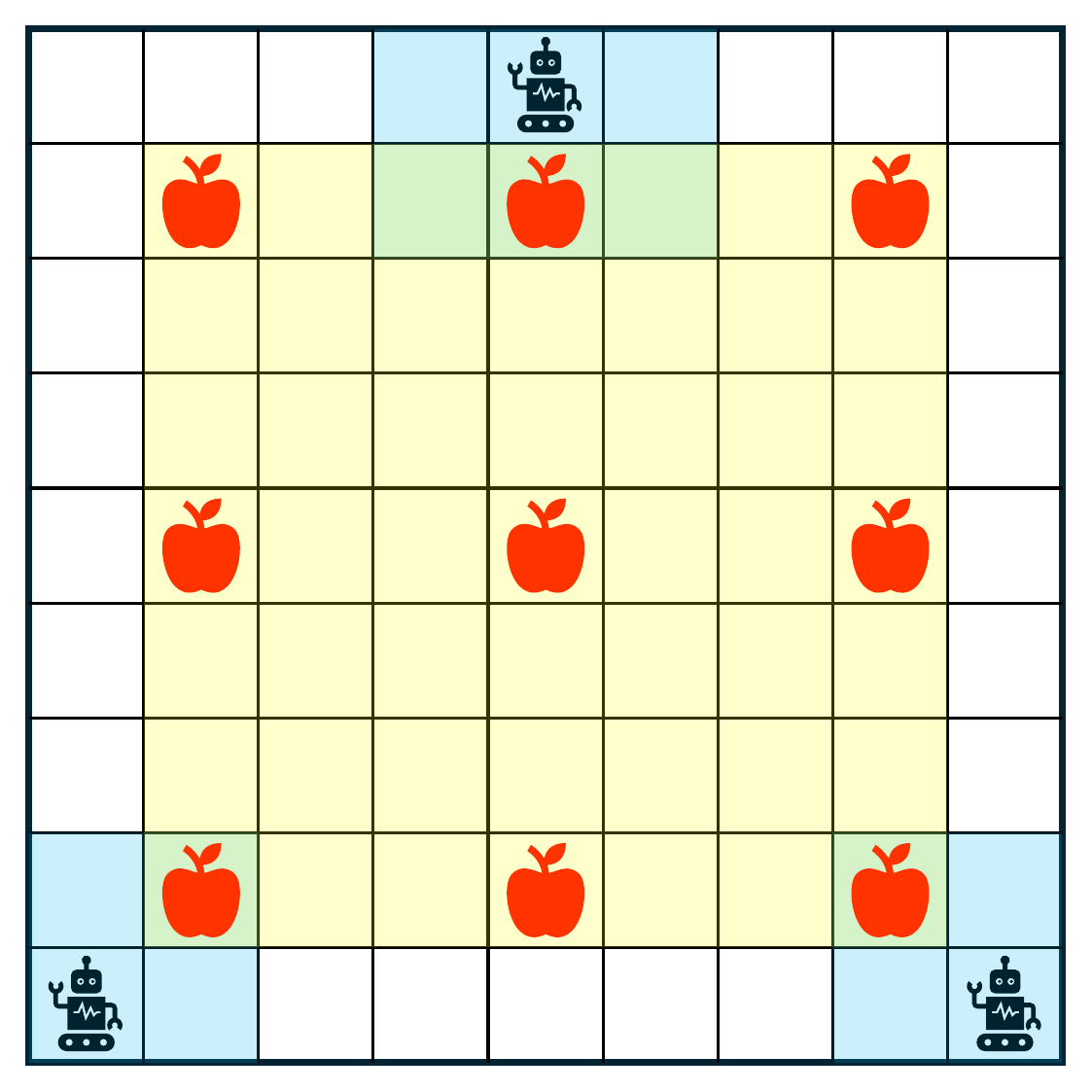}
\caption{Illustration of \textit{LBF: 9x9-20s}. The agents are required to cooperate within a limited number of time steps to concurrently collect the food based on their local observations. The blue areas indicate the agents’ local fields of view, the yellow areas represent possible spawn locations for the food (each corresponding to a specific task), and the red apples denote the food positions included in the training tasks.}
\label{fig:lbf}
\end{figure}
Level-Based Foraging (LBF)~\cite{lbf} is a widely used cooperative grid world environment in which agents must coordinate to collect food simultaneously. To accommodate the setting of rapid decentralized cooperative adaptation, we introduce the following modifications. Taking \textit{LBF: 9x9-20s} in Fig.~\ref{fig:lbf} as an example, three agents are initialized at fixed positions (as shown in the figure), each with a limited field of view covering only the adjacent grid cells (highlighted in blue). The task can only be successfully completed if all three agents simultaneously execute the foraging action around the sole food item on the map. Each agent’s discrete action space consists of six actions: staying still, moving in four directions, and foraging. The team receives a reward of 0.33 each time an agent first reaches the food, and an additional reward of 1 upon successful foraging. An episode terminates either when the agents complete the task or when the time step limit (20 steps in this scenario) is reached. In the figure, yellow areas indicate the possible spawn locations of the food, while red apples denote the nine fixed tasks used during training. We employ the QMIX~\cite{qmix} algorithm, implemented using the EPyMARL codebase, to collect 200 trajectories for each of the nine tasks at 0\%, 25\%, 50\%, 75\%, and 100\% of the maximum return, resulting in a multi-task dataset comprising 9,000 trajectories in total. For the alternative scenario, \textit{LBF: 7x7-15s}, we reduce the grid world size to 7x7 and correspondingly shorten the maximum time steps to 15.

\begin{figure*}[t]
\centering
\includegraphics[width=0.9\textwidth]{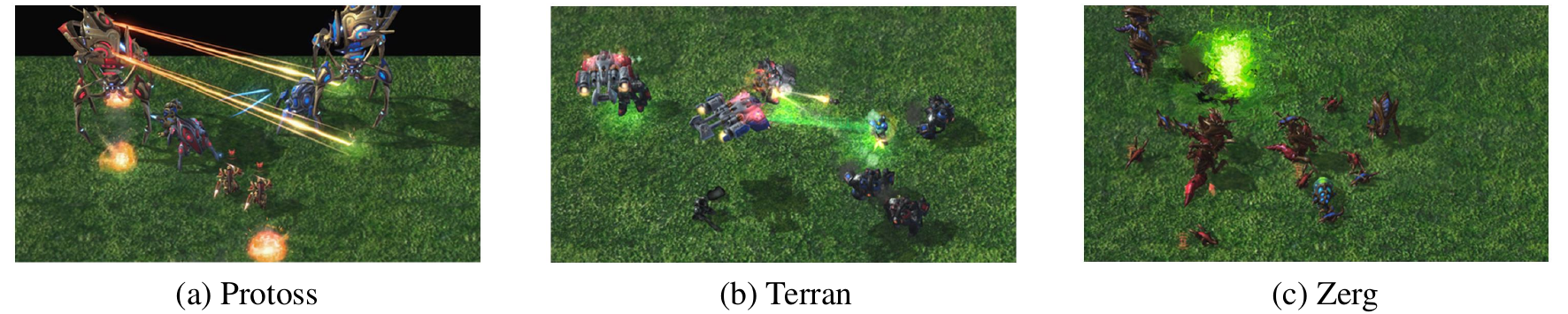} 
\caption{Illustration of SMAC. Agents are divided into three races: (a) Protoss, (b) Terran, and (c) Zerg. In each scenario, a random number of agents with randomly selected unit types from the chosen race fight against an equal number of enemy units controlled by the built-in AI.}
\label{fig:smac}
\end{figure*}

\subsection{StarCraft Multi-Agent Challenge}
StarCraft Multi-Agent Challenge (SMAC) is built upon the popular real-time strategy game StarCraft II and comprises a diverse set of cooperative multi-agent tasks. As illustrated in Fig.~\ref{fig:smac}, the game features three distinct races—Protoss, Terran, and Zerg—each characterized by unique attributes. Every race consists of two basic unit types and one exceptional unit type, each fulfilling different roles and capable of executing various tasks. Notably, the exceptional unit type can exert a substantial influence on the outcome of a battle, as listed in Tab.~\ref{tab:smac}. In each scenario, allied agents engage in combat against enemy agents controlled by the built-in AI. Each agent may choose from a discrete set of actions, including remaining stationary, moving in one of four directions, or attacking an enemy within its field of view. Team rewards are assigned based on the progression of the battle, with a complete victory yielding a reward of 20.

\begin{table}[t]
  \centering
  \begin{tabular}{c|cc|c}
    \toprule
    \textbf{Race} & \multicolumn{2}{c|}{\textbf{Basic Unit Types}} & \multicolumn{1}{c}{\textbf{Exceptional Unit Type}} \\
    \midrule
    \textbf{Protoss} & Stalker & Zealot & Colossus \\
    \midrule
    \textbf{Terran} & Marine & Marauder & Medivac \\
    \midrule
    \textbf{Zerg} & Zergling & Hydralisk & Baneling \\
    \bottomrule
  \end{tabular}
  \caption{The unit types of different races in SMAC.}
  \label{tab:smac}
\end{table}

Based on the above settings, we designed two types of experiments. The first is SMACv1~\cite{smac}, which exhibits relatively low randomness and comprises three scenarios: \textit{protoss}, \textit{terran}, and \textit{zerg}. In each scenario, tasks differ in both the number and types of agents (restricted to basic unit types), with fixed spawn positions for all agents. For each scenario, we consider three agent configurations: 5v5, 7v7, and 10v11. Using the QMIX algorithm, we collected 500 trajectories for each configuration at both 50\% and 100\% of the maximum return, resulting in a total of 3,000 trajectories per scenario. During testing, the number of allied agents is randomly set between 3 and 12, with the number of enemy agents adjusted accordingly.

In addition, under the more complex SMACv2~\cite{smacv2} setting, we designed the most challenging scenario, denoted as \textit{all}. In this scenario, the differences between tasks are not only reflected in the number and types of agents (including both basic and exceptional unit types), but also in their spawn positions, which vary across tasks. This requires agent teams to thoroughly explore in order to understand the current task. Furthermore, within a single task, any of the three aforementioned races may be present, posing an additional challenge for general decision-making models. In this scenario, the multi-task offline dataset contains 9,000 trajectories, collected in a manner similar to that used for the three scenarios described above.

\section{Implementation Details}
\label{App.D}
\subsection{Baselines}
Here, we introduce the baselines used in our experiments, including the Multi-Agent Decision Transformer (MADT), ICRL methods, a multi-task MARL method, and an ablated version of MAICC.

\textbf{MADT}~\cite{madt} extends Decision Transformer (DT)~\cite{dt} to multi-agent systems by performing offline training through sequential modeling. Although it introduces a transformer-based model into MARL, it does not consider collaborative adaptation to unseen tasks. While the method proposes optional online fine-tuning, the expensive and time-consuming gradient updates hinder its practical application. Therefore, in our experimental results, the performance curve of MADT appears as a horizontal line, since its performance does not improve with an increasing number of episodes. By comparing with MADT, we demonstrate the advantage of ICRL-based methods in incorporating cross-trajectory information into decision-making. In contrast, MADT can only make decisions based on the contextual information of the current episode, which prevents the model from understanding the environment and thus hinders multi-task generalization.

\textbf{AT}~\cite{at} feeds cross-trajectory contexts into transformer-based models as prompts, enabling agents to adapt more quickly during the testing phase. During training, multiple historical trajectories are sorted in ascending order of return to form an implicit chain-of-experience, aiming to leverage hindsight for performance improvement. This trial-and-error update method, which does not require gradient updates, has demonstrated impressive results on MuJoCo tasks~\cite{mujoco}. However, this approach does not perform retrieval over in-context trajectories; in complex scenarios, simply selecting the most recent trajectories with higher returns is insufficient for agents to understand the characteristics of the environment, thereby hindering rapid adaptation.

\textbf{RADT}~\cite{radt} is a recent ICRL work that introduces retrieval augmentation for the first time. It first utilizes a pre-trained DT model as an embedding model, and then performs trajectory retrieval based on the embeddings of action tokens. The retrieved similar trajectories are combined with the current input trajectory via cross-attention to assist decision-making, achieving superior adaptation in single-agent tasks compared to previous methods. However, in complex multi-agent scenarios, its coarse-grained encoder design faces challenges, leading to the retrieval of irrelevant trajectories. In addition, it does not include specific module designs to address the challenges of decentralized execution, which causes it to struggle in complex environments.

\textbf{HiSSD}~\cite{hissd} is a recent work that achieves generalizable offline multi-task MARL through learning both common and task-specific skills. Although it performs well on multi-agent task suites with small task discrepancies, it lacks an online adaptation module. As a result, when the differences between tasks are large, agents are unable to adapt to new tasks using the learned skills, which leads to a collapse in cooperation.

\textbf{MAICC-S} is an ablated version of our method, with the only difference being that it does not utilize the centralized training property. Instead, it optimizes the DEM using only the following loss function:
\begin{align}
&\mathcal{L}_{\mathrm{DEM}}=\mathcal{L}_{\mu}+\mathcal{L}_R+\mathcal{L}_{\mathcal{T}},\\
&\mathcal{L}_{\mu}=-\mathbb{E}_{\tau\sim\mathcal{D}}\sum_{h=0}^{H-1}\sum_{j=1}^n \log \mathrm{MLP}_{o\to a}(a^h_j|z_{o,j}^h),\\
&\mathcal{L}_R=\mathbb{E}_{\tau\sim\mathcal{D}}\sum_{h=0}^{H-1}\sum_{j=1}^n\left(r^h-\mathrm{MLP}_{a\to r}(z_{a,j}^h)\right)^2,\label{eq.maicc_s_r}\\
&\mathcal{L}_{\mathcal{T}}=-\mathbb{E}_{\tau\sim\mathcal{D}}\sum_{h=0}^{H-2}\sum_{j=1}^n\log\mathrm{MLP}_{p\to o}(o^{h+1}_j|z_p^h,o^h_j).
\end{align}
Here, the agent performs autoregressive prediction solely based on the embeddings generated from local information, and Eq.~\ref{eq.maicc_s_r} cannot provide credit assignment to support subsequent decentralized execution. These factors together result in its collaborative adaptation ability on unseen tasks being inferior to that of MAICC, especially in SMAC scenarios.

\subsection{Model Architecture}
Both the embedding models and the decision model in MAICC are implemented based on the GPT-2 model from the transformer codebase~\cite{huggingface}. For CEM, we introduce intra-team visibility, which is achieved by modifying the causal mask in the model. This allows each agent’s observation and action tokens at the same timestep to attend to information from teammates, enabling team-level trajectory modeling. In addition, for all baselines and ablation studies in our experiments (except for HiSSD, as it is not transformer-based), we use models of the same scale to ensure the fairness of the experimental results.

\subsection{Hyper-Parameter Settings}
In this paper, multiple hyperparameters are involved. Here, we list the key hyperparameter settings used in our experiments, as shown in Tab.~\ref{hyper-para}. 

\begin{table}[t]
  \centering
  \begin{tabular}{lc}
    \toprule
    \textbf{Attribute} & \textbf{Value} \\
    \midrule
    Embedding size for GPT-2 & $64$ \\ 
    Number of Layers for GPT-2 & $8$ \\ 
    Number of attention heads for GPT-2 & $8$\\
    The dropout rate for GPT-2 & $0.1$\\
    Hidden layers of the MLPs & $[64,64]$ \\
    Learning rate & $5e-4$\\
    Batch size & $32$ \\
    Number of in-context trajectories $k$ & $3/2$\\
    Number of online adaptation episodes $T$ & $200/100$\\
    Coefficient $\alpha$ & $0.8$\\
    Exponential decay rate $\lambda$ & $-\log 0.01$\\
    Optimizer & Adam \\
    \bottomrule
  \end{tabular}
  \caption{The key hyper-parameters in MAICC. The values before the $/$ correspond to the LBF tasks, while the values after the $/$ correspond to the SMAC tasks.}
  \label{hyper-para}
\end{table}

\subsection{Computing Infrastructure}
Most experiments were conducted on a server outfitted with an AMD EPYC 9654 96-Core Processor CPU, a NVIDIA GeForce RTX 4090 GPU, and 125 GB of RAM, running Ubuntu 20.04. MAICC was trained in a Python environment, and all relevant software libraries, along with their names and versions, are specified in the requirements.txt file included with the code.
\section{Additional Experiment Results}
\subsection{Visualization of Learned Embeddings}
\begin{figure}[t]
\centering
\includegraphics[width=0.9\columnwidth]{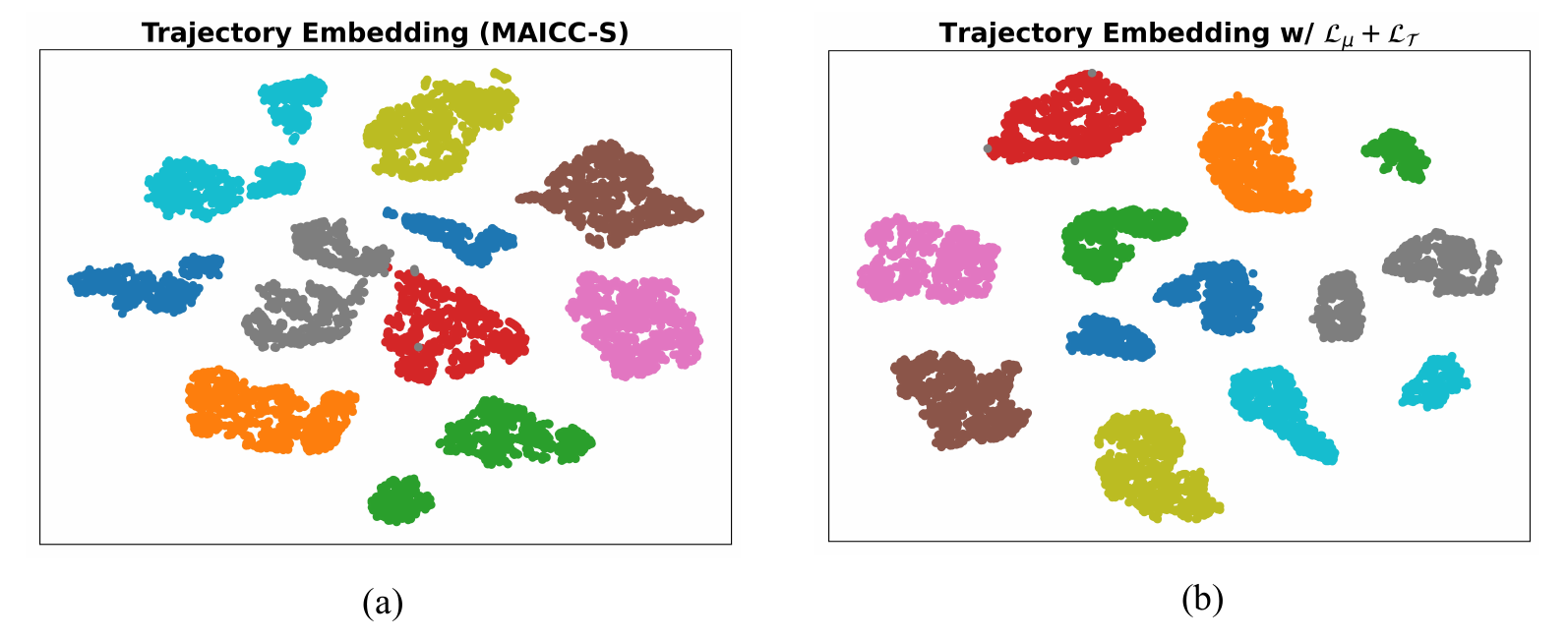}
\caption{Visualization results illustrating the effects of different embedding model training settings. Each point in the figure represents the embedding of a trajectory from the dataset, with points of the same color corresponding to trajectories from the same task.}
\label{fig:abl_emb}
\end{figure}
Here, similar to Sec.~5.3 in the main paper, we further conduct a visualization analysis of the learned trajectory embeddings under different settings, as shown in Fig.~\ref{fig:abl_emb}. In (a), we show the results when the team information distilled by CEM is not used during training. It can be observed that the blue trajectories form two distinct clusters, which will inevitably affect the effectiveness of trajectory retrieval. In (b), for the training results using only $\mathcal{L}_\mu$ and $\mathcal{L}_\mathcal{T}$, we also observe an overfitting phenomenon in the trajectory representations. These results, together with the experimental results in the main paper, demonstrate the effectiveness of our embedding model design.
\subsection{Ablation Study}
Similar to Sec. 5.4 in the main paper, we conducted the same ablation study on \textit{LBF: 9x9-20s}. The results are presented in Tab.~\ref{app:tab_abl}. It can be observed that, in the grid-world environment, our default setting still achieves the best performance, demonstrating the effectiveness of each module. Among them, the design of embedding models has a relatively minor impact on the results, whereas the constructed memory plays a crucial role in experimental performance. Relying solely on either the offline dataset or the online buffer for retrieval leads to severe adaptation failure.
\begin{table}[t]
  \centering
  \begin{tabular}{c|cccc|c}
    \toprule
    \textbf{Variants} & \multicolumn{1}{c}{\textbf{EM With RTG}} & \multicolumn{1}{c}{\textbf{Coefficient $\beta$}} & \multicolumn{1}{c}{\textbf{CEM loss}} & \multicolumn{1}{c|}{\textbf{Hyper-parameter $\alpha$}} &  \multicolumn{1}{c}{\textbf{\textit{LBF: 9x9-20s} Ret.}} \\
    \midrule
    \textbf{Default} & False & $\beta_t=\exp(-\lambda \frac{t}{T})$ & $\mathcal{L}_{\mu}+\mathcal{L}_R+\mathcal{L}_{\mathcal{T}}$ & $\alpha=0.8$ & 1.71$\pm$0.08\\ 
    \midrule
    (A)& True & & & & 1.69$\pm$0.09 \\
    \midrule
    \multirow{2}{*}{(B)}&  &  $\beta_t=0$ & & & 0.68$\pm$0.07\\
    & & $\beta_t=1$& & & 0.84$\pm$0.09\\
    \midrule
    \multirow{3}{*}{(C)}&  &  & $\mathcal{L}_{\mu}+\mathcal{L}_R$ & &1.58$\pm$0.11\\
    & & & $\mathcal{L}_{\mu}+\mathcal{L}_{\mathcal{T}}$ & $\alpha=1$ &1.60$\pm$0.12\\
    & & & $\mathcal{L}_{\mu}$ & $\alpha=1$ &1.44$\pm$0.09\\
    \midrule
    \multirow{2}{*}{(D)}& & & & $\alpha=1$ & 1.64$\pm$0.08\\
    & & & & $\alpha=0$ & 1.66$\pm$0.07\\
    \bottomrule
  \end{tabular}
  \caption{Ablation Study on MAICC. Unless otherwise noted, all settings follow the default configuration. “Ret.” indicates the average return over 50 test runs (with 95\% confidence interval), evaluated in the final adaptation episode.}
  \label{app:tab_abl}
\end{table}
\subsection{Sensitivity of Hyper-Parameters}
\begin{figure}[t]
\centering
\includegraphics[width=0.9\columnwidth]{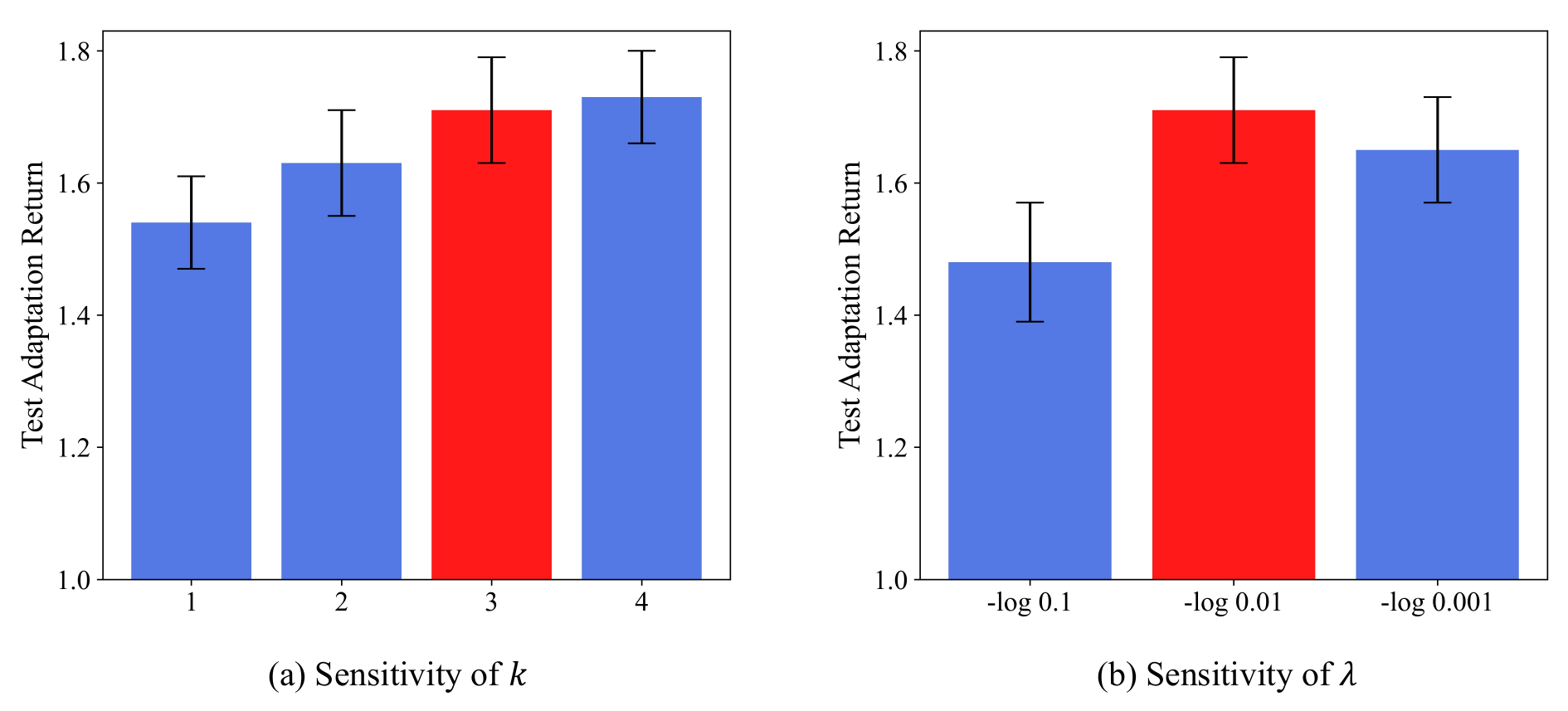}
\caption{Sensitivity of hyper-parameters. The red bars indicate the parameter values selected in our main experiments.}
\label{fig:sensi}
\end{figure}
In MAICC, the selection of certain hyper-parameters affects the experimental results. Therefore, in this subsection, we conduct a sensitivity analysis on two key hyper-parameters in the \textit{LBF: 9x9-20s} scenario. As shown in Fig.~\ref{fig:sensi}(a), we first conduct a sensitivity analysis on the number of in-context trajectories $k$. A larger $k$ allows the context to provide more information, enabling the agent to better understand the task requirements. However, since the inference speed of the transformer grows quadratically with the context length, we ultimately set $k=3$ for LBF environments and $k=2$ for SMAC tasks to balance performance and inference speed. In addition, Fig.~\ref{fig:sensi}(b) presents our study on the exponential decay rate $\lambda$. We observe that if the decay is too fast or too slow, the quality of trajectory retrieval deteriorates, which in turn reduces the final adaptation performance. Therefore, we set $\lambda = -\log 0.01$ in all experiments.
\end{document}